\documentclass[a4paper,11pt]{article}
\usepackage{amssymb,amsmath}
\usepackage{amsfonts,amssymb}
\usepackage{amsthm}
\usepackage{epic}
\usepackage{accents}
\usepackage{texdraw}
\usepackage{color}
\usepackage{graphicx,subfigure}
\usepackage{cite}
\usepackage[all]{xy}

\usepackage{xcolor}

\newcommand{\be}{\begin{equation}}
\newcommand{\ee}{\end{equation}}
\newcommand{\bea}{\begin{eqnarray}}
\newcommand{\eea}{\end{eqnarray}}
\newcommand{\bse}{\begin{subequations}}
\newcommand{\ese}{\end{subequations}}

\def \h#1{\widehat{#1}}
\def \t#1{\widetilde{#1}}
\def \b#1{\overline{#1}}

\def \tb#1{\widetilde{\overline{#1}}}
\def \th#1{\widehat{\widetilde{#1}}}
\def \hb#1{{\widehat{\overline{#1}}}}
\def \bt#1{\widetilde{\overline{#1}}}
\def \bh#1{{\widehat{\overline{#1}}}}

\def \thb#1{\widehat{\widetilde{\overline{#1}}}}

\def\m@th{\mathsurround=0pt}
\mathchardef\bracell="0365
\def\upbrall{$\m@th\bracell$}
\def\undertilde#1{\mathop{\vtop{\ialign{##\crcr
    $\hfil\displaystyle{#1}\hfil$\crcr
     \noalign
     {\kern1.5pt\nointerlineskip}
     \upbrall\crcr\noalign{\kern1pt
   }}}}\limits}

\newtheorem{thm}{Theorem}[section]

\newtheorem{prop}[thm]{Proposition}
\newtheorem{lem}{Lemma}[section]

\numberwithin{equation}{section}

\title{Integrability of auto-B\"acklund transformations,\\
and solutions of a torqued ABS equation}

\author{Xueli Wei$^{1}$, ~Peter H. van der Kamp$^{2}$,~
Da-jun Zhang$^{1}$\footnote{Corresponding author. Email: djzhang@staff.shu.edu.cn}\\
{\small  ${}^{1}$Department of Mathematics,
 Shanghai University, Shanghai 200444,  P.R. China}\\
 {\small ${}^{2}$Department of Mathematics and Statistics, La Trobe University, Victoria 3086, Australia}}

\date{\today}

\begin{document}

\maketitle

\begin{abstract}
An auto-B\"acklund transformation for the quad equation $\mathrm{Q1}_1$ is considered as a discrete equation,
called $\mathrm{H2}^a$, which is a so called torqued version of $\mathrm{H2}$.
The equations $\mathrm{H2}^a$ and $\mathrm{Q1}_1$ compose a consistent cube,
from which a auto-B\"acklund transformation and a Lax pair for $\mathrm{H2}^a$ are obtained.
More generally it is shown that auto-B\"acklund transformations admit auto-B\"acklund transformations.
Using the auto-B\"acklund transformation for $\mathrm{H2}^a$ we derive a seed solution and a one-soliton solution.
From this solution it is seen that $\mathrm{H2}^a$ is a semi-autonomous lattice equation,
as the spacing parameter $q$ depends on $m$ but it disappears from the plain wave factor.

\vskip 6pt

\noindent
\textbf{Key Words:}\quad auto-B\"acklund transformation, consistency, Lax pair, soliton solution, torqued ABS equation,
semi-autonomous.
\end{abstract}

\section{Introduction}  \label{sec-1}

The subtle concept of integrability touches on global existence and regularity of solutions, exact solvability,
as well as compatibility and consistency (cf.\cite{HJN-book-2016}).
In the past two decades, the study of discrete integrable system has achieved a truly significant
development, which mainly relies on the effective use of the property of multidimensional consistency (MDC).
In the two dimensional case, MDC means the equation is Consistent Around the Cube (CAC)
and this implies it can be embedded consistently into lattices of dimension 3 and higher
\cite{NW-GLA-2001,N-PLA-2002,BS-IMRN-2002}.
In 2003, Adler, Bobenko and Suris (ABS) classified scalar quadrilateral equations that are CAC (with extra restrictions:
affine linear, D4 symmetry and tetrahedron property)
\cite{ABS-CMP-2003}. The complete list contains 9 equations.

In this paper, our discussion will focus on two of them, namely
\begin{equation}
\mathrm{Q1}_{\delta}(u,\t u,\h u,\th u; p,q)=p(u-\h u)(\t u-\th u)-q(u-\t u)(\h u-\th u)+\delta pq(p-q)=0
\label{Q1-d}
\end{equation}
and
\begin{equation}
\mathrm{H2}(u,\t u,\h u,\th u; p,q)=(u-\th u)(\t u-\h u)+(q-p)(u+\t u+\h u+\th u)+q^{2}-p^{2}=0.
\label{H2}
\end{equation}
Here $u=u(n,m)$ is a function on $\mathbb{Z}^2$,
$p$ and $q$ are spacing parameters in the $n$ and $m$ direction  respectively,
$\delta$ is an arbitrary constant which we set equal to 1 in the sequel, and conventionally, tilde and hat denote shifts
\begin{equation}\label{notation}
 u=u(n,m), ~ \t u=u(n+1,m), ~ \h u=u(n,m+1), ~ \th u=u(n+1,m+1).
\end{equation}
H2 is a new equation due to the ABS classification, while Q1$_\delta$ extends the well known cross-ratio equation,
or lattice  Schwarzian Korteweg-de Vries equation Q1$_{\delta=0}$.
Note that spacing parameters $p$ and $q$ can depend on $n$ and $m$ respectively,
which  leads to nonautonomous equations.

For a quadrilateral equation that is CAC, the equation itself defines its own (natural) auto-B\"acklund transformation (auto-BT),
cf. \cite{ABS-CMP-2003}.
For example, the system
\[\mathrm{Q1}_{\delta}(u,\t u, \b u,\tb u; p,r)=0,~~
\mathrm{Q1}_{\delta}(u,\h u,\b u,\hb u; q,r)=0,
\]
where $r$ acts as a wave number, composes an auto-BT between $\mathrm{Q1}_{\delta}(u,\t u,\h u,\th u; p,q)=0$
and $\mathrm{Q1}_{\delta}(\b u,\tb u,\hb u,\thb u; p,q)=0$.
Such a property has been employed in solving CAC equations, see
e.g. \cite{AHN-JPA-2007,AHN-JPA-2008,HZ-JPA-2009,HZ-JMP-2010,HZ-SIGMA-2011}.

Some CAC equations allow auto-BTs of other forms. For example, in \cite{Atk-JPA-2008}
it was shown that the coupled system
\begin{subequations}\label{BT-Q11}
\begin{align}
A: &\ (u-\t u)(\bt u- u)-p(u+\t u+\b u+\bt u+p+2r)=0,\\
B: &\ (u-\h u)(\bh u- u)-q(u+\h u+\b u+\bh u+q+2r)=0
\end{align}
\end{subequations}
provides an auto-BT between
\begin{equation} \label{Q}
Q\ : \ \text{Q1}_1(u, \t u,\h u,\th u; p,q)=0
\end{equation}
and $\overline{Q}:\ \mathrm{Q1}_{1}(\b u,\bt u,\bh u,\thb u; p,q)=0$,
and, that H2 acts as a nonlinear superposition principle for the BT \eqref{BT-Q11}.
One can think of the auto-BT as equations posed on the side faces of a consistent cube with
$Q$ and $\overline{Q}$ respectively on the bottom and the top face, as in Figure \ref{fig-1}.
Here one interprets $\b u=u(n,m,l+1)$, and $r$ serves as a spacing parameter for the third direction $l$.
The superposition principle can be understood as consistency of a 4D cube, see \cite{RB13,ZVZ}.

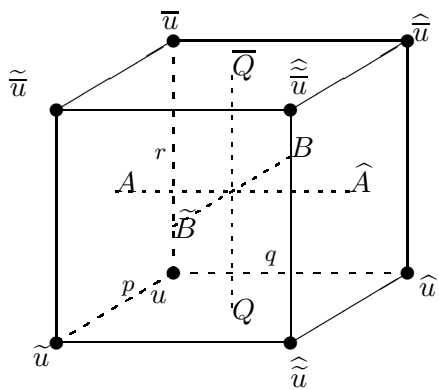
\begin{figure}[h]
\setlength{\unitlength}{0.08em}
\hspace{5cm}\begin{picture}(200,170)(10,-20)
  \put(100,  0){\circle*{6}} \put(0  ,100){\circle*{6}}
  \put( 50, 30){\circle*{6}} \put(150,130){\circle*{6}}
  \put(  0,  0){\circle*{6}}  \put(100,100){\circle*{6}}
  \put( 50,130){\circle*{6}}  \put(150, 30){\circle*{6}}
  \put( 0,  0){\line(1,0){100}}
  \put( 0,100){\line(1,0){100}}
  \put(50,130){\line(1,0){100}}
  \put(  0, 0){\line(0,1){100}}
  \put(100, 0){\line(0,1){100}}
  \put(150,30){\line(0,1){100}}
  \put(  0,100){\line(5,3){50}}
  \put(100,100){\line(5,3){50}}
  \put(100,  0){\line(5,3){50}}
   \dashline{3}(50,30)(0,0)
  \dashline{3}(50,30)(150,30)
\dashline{3}(50,30)(50,130)
   \dashline{2}(25, 65)(125, 65)
    \dashline{2}(75, 15)(75, 115)
      \dashline{2}(50, 50)(100, 80)
      \put(25, 65){$A$}
      \put(125, 65){$\h A$}
       \put(75, 10){$Q$}
       \put(75, 115){$\b Q$}
        \put(50, 45){$\t B$}
         \put(100, 80){$B$}
      \put(-10,-10){$\t u$}
     \put(100,-19){$ \th u $}
      \put(40,17){$ u $}
     \put(-20,106){$ \tb u $}
    \put(155,20){$ \h u $}
     \put(45,135){$ \b u  $}
     \put(100,107){$ \th{\b u} $}
    \put(152,130){$ \hb u $}
    \put(28,22){\footnotesize$p$}
\put(89,33){\makebox(0,0)[lb]{{\footnotesize$q$}}}
\put(42,79){\makebox(0,0)[lb]{{\footnotesize$r$}}}
\end{picture}
\caption{Consistent cube for $A, B$ and $Q$.
} \label{fig-1}
\end{figure}

In \cite{ZZV} the auto-BT \eqref{BT-Q11} and its superposition principle have been derived from the natural auto-BT for H2,
employing a transformation of the variables and the parameters. The equation
\begin{equation}\label{h2}
\begin{split}
\mathrm{H2}^a(u,\t u,\h u,\th u; p,q)&=
\mathrm{H2}(u,\th u,\h u,\t u; p+q,q)\\
&= (u-\t u)(\th u - \h u)-p(u+\t u+\h u+\th u+p+2q)=0
\end{split}
\end{equation}
was identified as a torqued version of the equation H2.
The superscript $^a$ refers to the {\em a}dditive transformation of the spacing parameter.
In \cite{Atk-JPA-2008} equation \eqref{h2} appeared as part of an auto-BT for Q$1_1$. 
The corresponding consistent cube is a special case of \cite[Eq. (3.9)]{Boll}. 
In \cite{ZZV} equation \eqref{h2} was shown to be an integrable equation in its own right, 
with an asymmetric auto-BT given by $A=\mathrm{H2}^a=0$ and $B=\mathrm{H2}=0$. 
Here we provide an alternative auto-BT for equation \eqref{h2} to the one that was provided in \cite{ZZV}.

In section \ref{sec-2}, we establish a simple but quite general result,
namely that if a system of equations $A=B=0$ comprises an auto-BT then both equations $A=0$ and $B=0$ admit an auto-BT themselves. In particular, the equation $\mathrm{H2}^a$ given by (\ref{h2}) is CAC,
with $\mathrm{H2}^a$ and $\mathrm{Q1}_{1}$ providing its an auto-BT.
We construct a Lax pair for $\mathrm{H2}^a$, which is asymmetric.
In section \ref{sec-3}, we employ the auto-BT for $\mathrm{H2}^a$ to derive a seed-solution and the corresponding one-soliton solution.
In the seed-solution the spacing parameter $q$ depends explicitly on $m$,
which makes $\mathrm{H2}^a$ inherent semi-autonomous.
Some conclusions are presented in section \ref{sec-4}.

\section{Auto-BTs for Auto-BTs, and a Lax pair for $\mathrm{H2}^a$}\label{sec-2}

To have a consistent cube with $\mathrm{H2}^a$ and $\mathrm{Q1}_{1}$ on the side faces,
providing an auto-BT for $\mathrm{H2}^a$, we assign equations to six faces as follows:
\begin{subequations}\label{6 eqs}
\begin{align}\label{2.1a}
& Q: \mathrm{H2}^a(u,\t u,\h u,\th u;p,q)=0, ~& \b Q: \mathrm{H2}^a(\b u,\tb u,\hb u,\th{\b u};p,q)=0,\\
& A: \mathrm{Q1}_{1}(u,\t u,\b u,\tb u;p,r)=0, ~& \h A: \mathrm{Q1}_{1}(\h u,\th u,\hb u,\th{\b u};p,r)=0,
\label{2.1b}
\\
\label{2.1c}
&B: \mathrm{H2}^a(u,\b u,\h u,\hb u;r,q)=0, ~& \t B: \mathrm{H2}^a(\t u,\tb u,\th u,\th{\b u};r,q)=0.
\end{align}
\end{subequations}
Then, given initial values  $u, \t u, \h u, \b u$, by direct calculation,
one can find that the value $\th{\b u}$ is uniquely determined.
Thus, the cube in Figure \ref{fig-1} with \eqref{6 eqs} is a consistent cube.

By means of such a consistency,
the side equations $A$ and $B$, i.e.
\begin{subequations}\label{bt}
\begin{align}
A:~~& p(u-\b u)(\t u-\tb u)-r(u-\t u)(\b u-\tb u)+pr(p-r)=0,\label{11}\\
B:~~& (u-\b u)(\hb u-\b u)-r(u+\b u+\h u+\hb u+r+2q)=0,\label{12}
\end{align}
\end{subequations}
compose an auto-BT for the $\mathrm{H2}^a$ equation \eqref{h2}.
Here $r$ acts as the B\"acklund parameter.

We note that the order of the variables in the equations (2.1) is quite particular.
Since equation \eqref{h2} is not D4 symmetric, i.e. we have
\[
\mathrm{H2}^a(u,\b u,\h u,\hb u;r,q)\neq \mathrm{H2}^a(u,\h u,\b u,\hb u;q,r),
\]
one has to be careful. The above result is explained by the following  useful result, cf. \cite[Section 2.1]{JH19} 
where the same idea was used to reduce the number of triplets of equations to consider for the classification of consistent cubes.

\begin{lem}\label{T-1}
Let
\begin{equation}
   A(u,\t u, \b u, \tb u ; p,r)=0,~~ B(u,\h u, \b u, \hb u; q,r)=0
\end{equation}
be an auto-BT for
\begin{equation}
   Q(u,\t u, \h u, \th u; p,q)=0.
\end{equation}
Then we have (i)
\begin{equation}\label{QB}
   Q(u,\t u, \b u, \tb u; p,r)=0,~~ B(u,\b u, \h u, \th u;r,q)=0
\end{equation}
is an auto-BT for
\begin{equation}\label{A}
   A(u,\t u, \h u, \th u; p,q)=0;
\end{equation}
and (ii)
\begin{equation}\label{QA}
   Q(u,\b u, \t u, \tb u;r,p)=0,~~ A(u,\b u, \h u, \bh u;r,q)=0
\end{equation}
is an auto-BT for
\begin{equation}\label{B}
   B(u,\t u, \h u, \th u; p,q)=0.
\end{equation}
\end{lem}
 \begin{proof}
If $A=B=0$ is an auto-BT of $Q=0$, then they compose a consistent cube as in Figure \ref{fig-1}.
We prove the result by relabelling the fields at the vertices, cf. \cite[Lemma 2.1]{ZVZ}.
For (i) we interchange $\h u \leftrightarrow \b u$ and $q \leftrightarrow r$,
and for (ii) we perform the cyclic shifts $\h u \rightarrow \t u \rightarrow \b u \rightarrow \h u$
and $q \rightarrow p \rightarrow r \rightarrow q$.
\end{proof}
Applying (i) to the consistent cube with \eqref{BT-Q11} and \eqref{Q} we obtain \eqref{6 eqs}.
Applying (ii) yields the same, as Q1$_1$ has D4 symmetry.\\


3D consistency can be used to construct Lax pairs for quadrilateral equations
(cf.\cite{N-PLA-2002,ABS-CMP-2003,BHQK-FCM-2013}).
To achieve a Lax pair for $\mathrm{H2}^a$, we rewrite \eqref{bt} as
\begin{subequations}\label{bt-alt}
\begin{align}
& \tb u=\frac{u (p\t u - r\b u) + (p - r) (pr - \t u \b u)}{(p - r) u + r\t u - p \b u},\\
& \hb u=-r + \h u - \frac{2 r (q + \h u + u)}{r - u + \b u}.
\end{align}
\end{subequations}
Then, introducing $\b u=G/F$ and $\varphi=(G,F)^T$, from \eqref{bt-alt} we have
\begin{equation}\label{lax}
\t \varphi=L\varphi,~~\h\varphi=M\varphi,
\end{equation}
where
\begin{align*}
&    L=\gamma\left ( \begin{array}{cc}
              -ur-(p-r)\t u             &  pu\t u+(p-r)pr            \\
              -p                              &  (p-r)u+r\t u             \\
                   \end{array}
           \right ),\\
& M=\gamma'\left ( \begin{array}{cc}
             \h u-r              &  (-r+\h u)(r-u)-2r(q+u+\h u)       \\
             1                      &  r-u                                    \\
                   \end{array}
           \right ),
\end{align*}
with $\gamma=\frac{1}{\sqrt{p^{2}-(u-\t u)^{2}}},~\gamma'=\frac{1}{\sqrt{q+u+\h u}}$.
The linear system \eqref{lax} is compatible for solutions of \eqref{h2}
in the sense that $\mathrm{H2}^a$ is a divisor of $(\h L M)^2-(\t M L)^2$,
where the square can be taken either as matrix multiplication, or as component-wise multiplication.

\section{Seed and one-soliton solution}\label{sec-3}
Our idea of constructing solutions for \eqref{h2} is to use its auto-BT \eqref{bt}.
First, we need to have a simple solution as a ``seed".
To find such a solution, we take $\b u=u$ in the BT \eqref{bt}, i.e.
\begin{equation}\label{bt-fix}
(u-\t u)^{2}=p(p-r),~~  u+\h u=-q-\frac{r}{2}.
\end{equation}
Such a treatment is called \textit{fixed point idea}, which has proved effective in finding seed solutions
\cite{AHN-JPA-2007,HZ-JPA-2009}.

\begin{prop}
Parametrising
\begin{equation}\label{pq}
p=\frac{\alpha}{a},\quad \alpha = -\frac{ac}{a^2 - 1},\quad q=(-1)^{m}\beta-\frac{c}{2},
\end{equation}
and setting the seed BT parameter equal to $r=c$, the equations \eqref{bt-fix} allow the solution
\begin{equation}\label{00}
u_0=(-1)^{m}(\alpha n+\beta m+c_{0})
\end{equation}
where $c_{0}$ is a constant.
\end{prop}
\begin{proof}
By direct calculation; with the given parametrisations the equations \eqref{bt-fix} read
\[
(u-\t u)^{2}=\alpha^2,~~ u+\h u=(-1)^{m}\beta.
\]
\end{proof}

\noindent It can be verified directly that \eqref{00} also provides a solution to \eqref{h2}.
Next, we derive the one-soliton solution for \eqref{h2}, from the auto-BT \eqref{bt} with $u=u_0$ as a seed solution.
\begin{prop}
The equation \eqref{h2}, with lattice parameters \eqref{pq}, admits the one-soliton solution
\begin{equation}\label{1ss}
u_1=(-1)^{m}\left(\alpha n+\beta m+c_{0}
+\frac{ck}{1-k^2}\frac{1-\rho_{n,m}}{1+\rho_{n,m}}\right),
\end{equation}
where
\begin{equation}
\rho_{n,m}=\rho_{0,0}\left(\frac{a+k}{a-k}\right)^{n}\,\prod_{i=0}^{m-1}{\frac{(-1)^{i}-k}{(-1)^{i}+k}}
\end{equation}
with constant $\rho_{0,0}$, is the plain wave factor.
\end{prop}
\begin{proof}
Let
\begin{equation}\label{u-bar}
u_{1}=u_{0}+(-1)^{m}(\kappa + \nu),
\end{equation}
where $\kappa=kr$. With \eqref{pq} and parametrising the first BT parameter by
\begin{equation}\label{02}
r=\frac{c}{1-k^{2}},
\end{equation}
then substitution of $u=u_0$ and $\b u=u_1$ into the auto-BT \eqref{bt} yields
\begin{equation}\label{qq}
\t \nu=\frac{\nu E_+}{\nu+E_-},~~
\h \nu=\frac{\nu F_+(m)}{\nu+F_-(m)},
\end{equation}
where
\begin{equation}\label{EF}
E_\pm=-r(a\pm k),~ F_\pm(m)=r((-1)^{m}\mp k).
\end{equation}
The difference system \eqref{qq} can be linearized using $\nu = \frac{f}{g}$ and $\Phi=(f,g)^{T}$, which leads to
\begin{equation}\label{Phi-eq}
\Phi(n+1,m)=M\Phi(n,m),~~ \Phi(n,m+1)=N(m)\Phi(n,m),
\end{equation}
where
\begin{equation}
    M= \left ( \begin{array}{cccc}
              E_+             &  0          \\
              1                             &  E_-
              \end{array}
           \right ),~~
            N(m)= \left ( \begin{array}{cccc}
             F_+              &  0      \\
             1                      & F_-
                   \end{array}
           \right ),
\end{equation}
By ``integrating" \eqref{Phi-eq} we have
\begin{equation}\label{Phi-eq2}
\Phi(n,m)=\mathcal{M}(n)\Phi(0,m),~~ \Phi(n,m)=\mathcal{N}(m)\Phi(n,0),
\end{equation}
where
\begin{equation*}
    \mathcal{M}(n)=\left ( \begin{array}{cccc}
             E_+^{n}             &  0          \\
              \dfrac{E_-^{n}-E_+^{n}}{2\kappa}         & E_-^{n}
                   \end{array}
           \right ),~~
    \mathcal{N}(m)=\left ( \begin{array}{cccc}
            \prod\limits_{i=0}^{m-1}F_+(i)              &  0      \\
             \dfrac{1-(-1)^{m}}{2}\prod\limits_{i=0}^{m-2}F_+(i)    & \prod\limits_{i=0}^{m-1}F_-(i)                                \\
                  \end{array}
           \right ).
\end{equation*}
Thus, we get a solution to \eqref{Phi-eq2}:
\begin{equation}
\Phi(n,m)=\mathcal{M}(n) \mathcal{N}(m)\Phi(0,0),
\end{equation}
from which  $\nu=f/g$ is  obtained as
 \begin{equation}\label{nu-1}
 \nu=\frac{E_+^{n}\prod\limits_{i=0}^{m-1}F_+(i)\cdot\nu_{0,0}}{E_-^{n}
 \prod\limits_{i=0}^{m-1}F_-(i)
 +\frac{\Bigl(E_-^{n}\prod\limits_{i=0}^{m-1}F_-(i)-E_+^{n}
 \prod\limits_{i=0}^{m-1}F_+(i)\Bigr)\nu_{0,0}}{2\kappa}},
 \end{equation}
where $\nu_{0,0}=\frac{f_{0,0}}{g_{0,0}}$.
Introducing the plain wave factor
\begin{equation}\label{315}
\rho_{n,m}=\rho_{0,0}\left(\frac{E_+}{E_-}\right)^{n}\,\prod_{i=0}^{m-1} \frac{F_+(i)}{F_-(i)}
\end{equation}
with constant $\rho_{0,0}$, the above $\nu$ is written as
\begin{equation} \label{nu}
\nu=\frac{-2\kappa\rho_{n,m}}{1+\rho_{n,m}}
\end{equation}
where some constants are absorbed into $\rho_{0,0}=\frac{-\nu_{0,0}}{2\kappa+\nu_{0,0}}$.
Substituting \eqref{nu} into \eqref{u-bar} yields the one-soliton solution \eqref{1ss},
which solves \eqref{h2} with \eqref{pq} and \eqref{02}.
Note that in the plain wave factor \eqref{315} $n,m\in \mathbb{Z}$,
and when $m \leq 0$ the product $\prod_{i=0}^{m-1} (\cdot )$ is considered as $\prod_{i=m-1}^{0} (\cdot )$.
\end{proof}

It is interesting that the solution has an oscillatory factor $(-1)^m$ in $m$-direction
and  in the plain wave factor $\rho_{n,m}$ the spacing parameter $q$ for $m$-direction does not appear.
Considering the parametrization \eqref{pq} where $p$ is constant while $q$ depends on $m$,
we can say that the $\mathrm{H2}^a$ equation \eqref{h2} is semi-autonomous.

\section{Conclusions}\label{sec-4}
In this paper, we have shown that equations which constitute an auto-BT for a quad equation admit auto-BTs themselves.
We have focussed on one such equation, the torqued H2 equation denoted $\mathrm{H2}^a$ \eqref{h2},
which forms an auto-BT for Q1$_1$. This equation is not part of the ABS list of CAC quad equations,
as it is not symmetric with respect to $(n,p) \leftrightarrow (m,q)$.
The integrability of this equation is guaranteed as it is part of a consistent cube, cf. \cite{ZZV}. The equations $\mathrm{H2}^a$ and Q1$_1$ comprise an auto-BT from which a Lax pair was obtained. Using this auto-BT we have derived a seed solution and a one-soliton solution.
The parametrisation of these solutions show that  $\mathrm{H2}^a$ is a semi-autonomous equation. We hope to be able to construct higher order soliton solutions in a future paper.

\subsection*{Acknowledgments}
This work was supported by a La Trobe University China studies seed-funding research grant, and by the NSF of China [grant numbers 11875040 and 11631007].

\end{document}